\documentclass[pdflatex,sn-mathphys-num]{sn-jnl}% Math and Physical Sciences Numbered Reference Style

\usepackage[utf8]{inputenc}

\usepackage[english]{babel}
\usepackage{ulem}
\usepackage{tikz}
\usepackage{hyperref}
\usepackage[plain]{fancyref}
\usepackage[T1]{fontenc}
\usepackage{csquotes}

\usepackage{graphicx}%
\usepackage{multirow}%
\usepackage{amsmath,amssymb,amsfonts, tabularx}%
\usepackage{amsthm}%
\usepackage{mathrsfs}%
\usepackage[title]{appendix}%
\usepackage{xcolor}%
\usepackage{textcomp}%
\usepackage{manyfoot}%
\usepackage{booktabs}%
\usepackage{algorithm}%
\usepackage{algorithmicx}%
\usepackage{algpseudocode}%
\usepackage{listings}%
\usepackage{natbib}
\usepackage{mathtools}

\setcitestyle{authoryear,open={(},close={)}} %Citation-related commands

\usetikzlibrary{decorations.pathreplacing}

\theoremstyle{thmstyleone}%
\newtheorem{theorem}{Theorem}%  meant for continuous numbers
\newtheorem{propositionCounted}{Proposition} % Different numbering for theorems and propositions

\theoremstyle{thmstyletwo}%

\theoremstyle{thmstylethree}%
\newtheorem{definition}{Definition}%

\newcommand{\Si}{\mathcal{S}}
\newcommand{\Ui}{\mathcal{U}}

\def\diamond{\mathrel{%
  \ooalign{$$\cr\hss\lower.255ex\hbox{\Bigmath\char5}\hss}}} 

\def\diamondplus{\mathrel{%
  \ooalign{$+$\cr\hss\lower.255ex\hbox{\Bigmath\char5}\hss}}} 
  
 \def\diamondminus{\mathrel{%
  \ooalign{$-$\cr\hss\lower.255ex\hbox{\Bigmath\char5}\hss}}} 

\def\since{\mathcal{S}}
\def\until{\mathcal{U}} 

\font\Bigmath=cmsy10 scaled \magstep2

\title{On Syntactical Simplification of Temporal Operators in Negation-free MTL}
\date{\today}

\author*[1]{\fnm{Mathijs} \sur{van Noort}}\email{Mathijs.vanNoort@UGent.be}

\author[1]{\fnm{Femke} \sur{Ongenae}}\email{Femke.Ongenae@UGent.be}
\author[1,2]{\fnm{Pieter} \sur{Bonte}}\email{Pieter.Bonte@KULeuven.be}

\affil*[1]{\orgdiv{IDLab}, \orgname{Ghent University - imec},
\orgaddress{
\country{Belgium}}}

\affil[2]{\orgdiv{Department of Computer Science}, \orgname{KU Leuven Campus Kulak}, \orgaddress{
\country{Belgium}}}

\begin{document}

\maketitle

\section*{Abstract}\label{sec:Abstract}

Temporal reasoning in dynamic, data-intensive environments increasingly demands expressive yet tractable logical frameworks. Traditional approaches often rely on negation to express absence or contradiction. In such contexts, Negation-as-Failure is commonly used to infer negative information from the lack of positive evidence. However, open and distributed systems such as IoT networks or the Semantic Web Negation-as-Failure semantics become unreliable due to incomplete and asynchronous data. This has led to a growing interest in negation-free fragments of temporal rule-based systems, which preserve monotonicity and enable scalable reasoning.

This paper investigates the expressive power of negation-free MTL, a temporal logic framework designed for rule-based reasoning over time. We show that the ``always'' operators $\boxplus$ and $\boxminus$, often treated as syntactic sugar for combinations of other temporal constructs, can be eliminated using ``once'',  ``since'' and ``until'' operators. Remarkably, even the ``once'' operators can be removed, yielding a fragment based solely on ``until'' and ``since''. These results challenge the assumption that negation is necessary for expressing universal temporal constraints, and reveal a robust fragment capable of capturing both existential and invariant temporal patterns. Furthermore, the results induce a reduction in the syntax of MTL, which in turn can provide benefits for both theoretical study as well as implementation efforts.

% Theoretical implications include a deeper understanding of minimal operator requirements for temporal expressivity, while practical benefits point to more efficient and predictable reasoning systems. This work opens avenues for optimizing temporal logic frameworks and adapting them to real-world applications where negation is impractical or undesirable. 

\section{Introduction}\label{sec:Intro}
% \input{TempDatalog/Introduction_4}
% \section{Intro New}
% \input{TempDatalog/Introduction5}

Negation is a widely used construct in formal, rule-based reasoning frameworks. It allows for the expression of absence, exclusion, or contradiction, and is often relied upon in a variety of application domains. For instance, in financial systems~\citep{baldazzi2023reasoning}, sensor monitoring~\citep{doherty2015hdrc3}, or E-health services~\citep{decontext}, the ability to express that something did not occur, or is not present, is often necessary. Without some form of negation, it becomes difficult to formally capture notions such as a no-show guest, a missing data point, or the absence of a required condition.

In logical rule-based systems, a common way to support negation is through the Closed World Assumption (CWA)~\citep{abiteboul_foundations_db_1995}. Under this assumption, if a fact is not explicitly known to be true, it is assumed to be false. This type of negation, referred to as Negation-as-Failure (NAF) is often a reasonable simplification in systems where the data is centrally managed and assumed to be complete~\citep{abiteboul_foundations_db_1995}.

However, recent developments in information systems have shifted attention toward open and distributed settings, where the CWA becomes harder to maintain~\citep{hayes_2001, hochstenbach2024rdf}. Examples include the World Wide Web, where data is distributed across many sources, and Internet of Things (IoT) environments~\citep{marinier2015maintaining, tu_iot_2020}, where components may be numerous, heterogeneous, and subject to change. In such settings, assuming that all relevant information is available and up-to-date at all times is increasingly difficult to justify. This poses particular challenges for negation: the absence of a fact in one part of the system does not necessarily imply its global absence. For instance, a missing sensor reading may be due to transmission delay, device failure, or simply a temporary disconnection, rather than a true absence of the measured phenomenon. Similarly, in federated or decentralized systems, the lack of a record in one data source does not guarantee that the record does not exist elsewhere. These ambiguities undermine the reliability of NAF semantics, which depend on a complete and static view of the data. Moreover, the dynamic and asynchronous nature of distributed systems -- where data sources may evolve independently -- complicates the interpretation of negation and calls for more cautious reasoning strategies that take these dynamic and asynchronous aspects into account. As a result, frameworks for open systems often require explicit mechanisms to handle partial knowledge and delayed information, further challenging the traditional logical foundations of negation.

At the same time, the growing availability of real-time data 
has led to increased interest in temporal rule-based reasoning~\citep{bonte_languages_2025, bonte_grounding_2024}. Many phenomena of interest are not characterized by isolated facts, but by patterns that unfold over time. For example, ~\cite{santipantakis2018stream} observes the spatio-temporal positions of vessels in order to detect various types of vessel activities, and \cite{bellomarini2025temporal} manages and evaluates evolutions in financial markets, most notably company ownership.
This has motivated the development of temporal extensions to classical logic-based systems, based on either Linear Temporal Logic  (LTL) or Metric Temporal Logic (MTL). Some recent developments in this field include Linear Temporal Public Announcement Logic (LTPAL)~\citep{dehkordi_linear_2020}, Metric Spatio-Temporal Logic (MSTL)~\citep{leng_qualitative_2016}, LARS~\citep{beck_lars_2018} and DatalogMTL~\citep{brandt_querying_2018}.

In addition to the considerations of more open and distributed systems, negation also presents specific challenges in streaming environments, where data arrives incrementally and the system must reason over a continuously evolving state~\citep{della_valle_its_2009}. In such settings, the absence of a fact at a given moment does not necessarily imply its permanent absence -- it may simply not have arrived yet. This makes it difficult to apply classical NAF semantics, which rely on a fixed and complete dataset. To address this, various stratified or windowed negation approaches have been proposed~\citep{walega_datalogmtl_2021, ketsman_datalog_2020, cucala_stratified_2021}, where negation is restricted to certain layers of the program or applied only within bounded time intervals. These are, however, constructions that inherently treat negation differently, with good reason: entailment in First-order Logic (FOL) -- and consequently in many rule-based reasoning systems -- is undecidable~\citep{trakhtenbrot_recursive_1953}. By introducing negation, either via NAF or stratified negation, a reasoning framework risks undecidability and may compromise its reliability, especially when reasoning process must be both timely and robust in the face of incomplete information.

These challenges have led to a growing interest in identifying expressive yet tractable fragments of temporal rule-based systems that avoid negation altogether. In particular, researchers have explored `positive fragments' -- systems that exclude negation to preserve monotonicity and decidability. For example,~\cite{urbani2022chasing} study positive and stratified LARS programs, while work on LTL and MTL often focuses on negation-free fragments to enable efficient model checking~\citep{abiteboul_foundations_db_1995, bouyer_expressiveness_2010, basin_monitoring_2015}.

While such fragments are intuitively less expressive -- since negation allows us to express opposites and contradictions -- they remain able to express the same temporal expressions of their larger, negation-including counterparts: A recurring theme in temporal logic is the equivalence between constructs like ``always P'' and ``not once not P.'' For instance, in LTL, this is expressed as $\mathbf{G} \phi \equiv \neg \mathbf{F} \big( \neg \phi \big)$~\citep{finger_adding_1992}, and in LARS,   $ \Box \alpha \equiv \neg \diamond  \neg \alpha $~\citep{beck_lars_2015}. These equivalences suggest that temporal ``always'' operators can be syntactic sugar for expressions involving negation and ``once''.

This raises an important question: what happens to such constructs when negation is removed? Must we reintroduce ``always'' as a core operator of the syntax? And if we do not, does this limit the expressive power of the system?

\textbf{In this work, we show that in MTL}, negation is not indispensable in the syntax: \textbf{the ``always'' operators $\boxplus$ and $\boxminus$ can be eliminated despite the lack of negation}.
This result indicates that a meaningful class of temporal expressions -- specifically, those that describe properties hold continuously over  intervals -- can be expressed within negation-free fragments of MTL. This shows that the expressivity of negation-free MTL is not as limited as intuition might suggest. They can capture not only existential or event-based patterns, but also universal temporal expressions. Furthermore, it shows that even in negation-free fragments, MTL over bounded intervals only requires temporal operators $\since$ and $\until$.

From a practical perspective, the results imply that negation-free MTL frameworks only need to account for two temporal operators, compared to six. This significant reduction of syntactic complexity opens the door towards a slimmer and more maintainable codebase. 
% by avoiding negation, rule evaluation becomes monotonic, meaning that adding new data cannot invalidate previous conclusions.
By showcasing how certain uses of negation can be eliminated through rewriting without sacrificing expressivity, we illustrate a viable negation-free fragment of MTL exists, which adequately balances expressivity with achieving a scalable
reasoning system. 

The remainder of this paper is structured as follows;
We first familiarize the reader with the syntax and semantics of DatalogMTL, as defined in the literature, in Section~\ref{sec:background}. The theoretical rewritings are presented in Section~\ref{sec:Theorems}. We discuss the implications of our results in Section~\ref{sec:conclusion}.

% \section{Syntax}\label{sec:Syntax}
% \input{TempDatalog/Syntax_TempDatalog}

% \section{Semantics}
% \label{sec:Semantics}
% \input{TempDatalog/Entailment}

\section{Preliminaries}\label{sec:background}

MTL, introduced by~\cite{koymans_specifying_1990}, introduced a framework for quantitative temporal reasoning in real-time systems. It has since found its way into multiple reasoning systems that deal with temporal data or changing states, such as DatalogMTL~\citep{brandt_querying_2018}, Predictive MTL (P-MTL)~\citep{tiger2016stream} and Metric Equilibrium Logic (MEL)~\citep{cabalar_metric_2022}. We briefly recall the syntax and semantics as it has been adopted throughout the literature~\citep{koymans_specifying_1990, alur_benefits_1996, brandt_querying_2018}. Within the scope of this work, we opt for a model-theoretic semantics and   rational timeline $(\mathbb{Q}, <)$.

\begin{definition}\label{def:MTLmodel}
    Given a list of predicate symbols $\mathcal{P}$,
    a \textit{model} $M$ is a triple $(\mathbb{Q}, <, V)$ where $V$ is a mapping from time points in $\mathbb{Q}$ to a subset of predicates in $\mathcal{P}$.
\end{definition}

Within the scope of this work, we consider Bounded MTL (BMTL), which only considers bounded intervals in the subscript of temporal operators:

\begin{definition}\label{def:MTLformula}
    The set of temporal formulae is defined recursively as follows:
    $$ A \coloneqq p | \top | \neg A | A \land A | \boxplus_I A |\boxminus_I A | \diamondplus_I A | \diamondminus_I A | A \since_I A | A \until_I A $$

    where $I$ is an positive bounded interval over interval over T, i.e. $I = [i_1, i_2]$ with $i_1, i_2 \in T$ and $0 \leq i_1 \leq i_2$.
\end{definition}

BMTL can be considered a fragment of ``full'' MTL, which also considers intervals of the form $[0, +\infty[$ in the subscript. A summary of BMTL and several other fragments and respective theoretical properties can be found in~\cite{ouaknine_some_2008}.

Following a model-theoretic approach, we define the semantics of temporal formulae $A$ by Table~\ref{tab:semantics}, with $M$ an interpretation, $t$, $t'$ and $t''$ time points on timeline $T$ and $I$ a positive bounded interval within $T$.

\begin{table}[ht]
\centering
\renewcommand{\arraystretch}{1.4}
\begin{tabularx}{\textwidth}{@{}lX@{}}
\textbf{Atom} & \textbf{Semantics} \\
\hline
Predicate & $M, t \models p $ if $p \in V(t)$\\
True & $M, t \models \top$ for each $t$ \\
Negation & $M, t \models \neg A $ iff $M, t \models A $ does not hold\\
Conjunction & $M, t \models A_1 \land A_2$ iff  $M, t \models A_1 $ and  $M, t \models A_2$  \\
% False & $M, t \models \bot$ for no $t$ \\ --> No longer defined
Once (past) & $M, t \models \diamondminus_I A$ iff $M, t' \models A$ for at least one $t'$ with $t - t' \in I$ \\
Once (future) & $M, t \models \diamondplus_I A$ iff $M, t' \models A$ for at least one $t'$ with $t' - t \in I$ \\
Always (past) & $M, t \models \boxminus_I A$ iff $M, t' \models A$ for all $t'$ with $t - t' \in I$ \\
Always (future) & $M, t \models \boxplus_I A$ iff $M, t' \models A$ for all $t'$ with $t' - t \in I$ \\
Since & $M, t \models A_1 \mathbin{\mathcal{S}_I} A_2$ iff $M, t' \models A_2$ for at least one $t'$ with $t - t' \in I$ and $M, t'' \models A_1$ for all $t'' \in [t', t]$ \\
Until & $M, t \models A_1 \mathbin{\mathcal{U}_I} A_2$ iff $M, t' \models A_2$ for at least one $t'$ with $t' - t \in I$ and $M, t'' \models A_1$ for all $t'' \in [t, t']$ \\
\end{tabularx}
\caption{Semantics of formulae}
\label{tab:semantics}
\end{table}

What can occur is that two formulae appear to be different, but have the same semantics. An example thereof is $\boxplus_{[0, i]} A$ and $A \until_{[i, i]} \top$; the first says $A$ must hold from $t$ (now) up to $t + i$. The latter says $A$ must hold from $t$ up until when $\top$ holds in interval $[t+i, t+i]$. Since $\top$ holds for any time point, it also does so for $t+i$, meaning the latter states $A$ holds from $t$ up until $t+i$.
We refer to such formulae with identical semantics as \textit{equivalent}, formally defined as follows:

\begin{definition}\label{def:equivalence_metric_atoms}
    Two formulae $A_1$ and $A_2$ are considered to be \textit{equivalent} if (and only if) for every model $M$ and for every time point $t$, it holds that $A_1$ is true at $t$ in $M$ if $A_2$ it true at $t$ in $M$ and vice versa. We denote $A_1 \equiv A_2$
\end{definition}

Table~\ref{tab:semantics} does not specify the exact interpretation for ``$M,t \models A$ does \textit{not} hold''. Each system can decide what is understood as ``it does not hold''. As considered in the Introduction, closed systems often employ a Negation-as-Failure approach; when there is no conclusive evidence in favor of $M, t \models A$, it assumes $M, t \models A$ does not hold. Imagine for example the reservation list of a restaurant; if there is no evidence you do have a reservation, then the waiter assumes you do not have a reservation.
NAF is not always a desired approach, since it may lead to premature conclusions. Perhaps you do have a reservation, but the waiter's list is outdated. In more open-ended systems, e.g . IoT or the Web, it is preferable to only say $M, t \models A$ does not hold when there is explicit evidence that says so. For example, you may not have a reservation for eight people, since there is specific evidence against it: your reservation is for two people. 

% \section{Temporal operators}
% \label{sec:TempOperators}
% \input{TempDatalog/TemporalOperators_Definitions.tex}

\section{Rewriting operators}\label{sec:Theorems}

As indicated in Section~\ref{sec:Intro}, rewriting operators in function of other operators is a known practice. For example,~\cite{gutierrez-basulto_metric_2016}
reiterates the syntax of $\text{LTL}^{\text{bin}} _{\mathcal{ALC}}$, where operations $\diamond_I C$ and $\Box_I C$ are introduced merely as a shorthand for $\top \mathcal{U}_I C$ and $\neg\big(\diamond_I \neg C \big)$. Similarly,~\cite{finger_adding_1992} define their respective temporal operators solely in function of $\mathcal{U}$ and $\mathcal{S}$, i.e.~$\mathbf{F}( A ) := \mathcal{U}(A, \top)$ and $\mathbf{P} (A) := \mathcal{S}( A, \top)$ for propositional temporal logics. \cite{koymans_specifying_1990} notes the same relations hold in MTL: $\top \until_I A \equiv \diamondplus_I A$ and $\top \since_I A \equiv \diamondminus_I A$. These equivalencies migrate to other frameworks as well, as seen in DatalogMTL~\citep{brandt_querying_2018}. These results can be summarized as follows:

\begin{propositionCounted}\label{prop:diamond_si_ui}
    Suppose $t\in T$, $I$ a non-negative interval and $A$ a formula. Then $ \diamondplus_I A \equiv \top \Ui_I A$ and $\diamondminus_I A \equiv \top \Si_I A$.
\end{propositionCounted}

In~\cite{gutierrez-basulto_metric_2016}, the MTL operator $\boxplus A$ is defined via the equivalence $\boxplus A \equiv \neg (\diamondplus \neg A)$. This formulation relies on the presence of negation to express ``always $A$''. However, in the absence of negation, such an equivalence cannot be reproduced within MTL, which raises a natural question: if we exclude negation $\neg$ \textit{and} the ``always'' operators $\boxplus$ and $\boxminus$ from the syntax of MTL, does this result in a strictly less expressive fragment compared to the full language?

The following theorems and proofs demonstrate that the answer is negative: even without these operators, the ability of MTL to express temporal relations remains intact. This finding challenges the intuition that negation is essential for expressing universal temporal properties and shows that such constructs can be captured through alternative means.

\begin{theorem}\label{th:AlwaysAfter2}
    Suppose $ t$ a time point, $I = [ i_1, i_2]$ a non-negative interval and $A$ a formula. Then 

    \begin{align*}
        \boxplus_I A &\equiv \diamondplus_{[i_1, i_1]} \Big( A \mathcal{U}_{[ i_2 - i_1, i_2 - i_1]} \top \Big)
        % \\ &
        \equiv \top \mathcal{U}_{[i_1, i_1]} \Big( A \mathcal{U}_{[ i_2 - i_1, i_2 - i_1]} \top \Big).
    \end{align*}
    
\end{theorem}

\begin{proof}
    Consider an arbitrary model $M$ and time point $t$ 

    \begin{align*}
        & && M, t \models \diamondplus_{[i_1, i_1]} \Big( A \mathcal{U}_{[ i_2 - i_1, i_2 - i_1]} \top \Big) \\
        &\Leftrightarrow &&\exists t' \in [ t + i_1 , t + i_1] : M, t' \models A \mathcal{U}_{[ i_2 - i_1, i_2 - i_1]} \top \\
        &\Leftrightarrow &&\exists t' \in [ t + i_1 , t + i_1] : 
        \\ & &&
        \Big( \exists t^{(2)} \in [ t' + (i_2 - i_1),  t' + (i_2 - i_1) ] :
        % \\ & &&
        ( M, t^{(2)} \models \top ) \land ( \forall t^{(3)} \in [ t', t^{(2)}] : M, t^{(3)} \models A) \Big)\\
        &\Leftrightarrow && \exists t^{(2)} \in [ t + i_2 ,  t + i_2 ] : 
        % \\ & &&
        \Big( ( t^{(2)} \models \top ) \land ( \forall t^{(3)} \in [ t + i_1, t^{(2)}] : t^{(3)} \models A) \Big)\\
        &\Leftrightarrow && \forall t^{(3)} \in [ t + i_1, t + i_2 ] : t^{(3)} \models A \\
        &\Leftrightarrow && t\models \boxplus_{I} A
    \end{align*}
    \vspace{2pt}

    The proof works as follows. In order for the equivalence to hold, the two formulae  must hold (and not hold) simultaneously for every model and time point. If this can be proven for an arbitrary pair of model and time point, it must hold for every pair. Hence, we consider an unspecified model $M$ and time point $t$ for which $M, t \models \diamondplus_{[i_1, i_1]} \big( A \Ui_{[i_2 - i_1, i_2 - i_1]} \top \big) $. Following the semantics of Section~\ref{sec:background}, there exists a future time point $t'$, such that $t'-t  \in [i_1,i_1]$, in other words $ t' = t+ i_1$, for which $M, t' \models A \Ui_{[i_2 - i_1, i_2 - i_1]} \top$. By the semantics of the $\until$ operator, this is equivalent to first the existence of a time point $t^{(2)}$ for which $t^{(2)} - t' \in [i_2 - i_1, i_2 - i_1]$, in other words $t^{(2)} = t' + (i_2 - i_1)$. By substituting the fact that $ t' = t+ i_1$, we get $t^{(2)} = t + i_2$. For this $t^{(2)}$ $M , t^{(2)} \models \top$ holds. Second, it entails that for every $t^{(3)} $ between $t'$ and $t^{(2)}$, it holds that $M, t^{(3)} \models A $. Since $\top$ holds for any $M$ and $t^{(2)}$ and $t^{(2)}$ can only be $t + i_2$, this is equivalent to $M, t^{(3)} \models A$ for every $t^{(3)}$ in interval $[t + i_1, t + i_2]$. In other words, $M, t \models \boxplus_{[i_1, i_2]} A$. It is easily verified each step likewise holds in the opposite direction.
    Lastly, the second equivalence follows from Proposition~\ref{prop:diamond_si_ui} by substituting the $\diamondplus$ operator.
\end{proof}

Analogously, we obtain an analogous result for $\boxminus$:

\begin{theorem}\label{th:AlwaysBefore2}
    Suppose $ t$ a time point, $I = [ i_1, i_2]$ a non-negative interval and $A$ a formula. Then 
    \begin{align*}
        \boxminus_I A &\equiv \diamondminus_{[i_1, i_1]} \Big( A \mathcal{S}_{[ i_2 - i_1, i_2 - i_1]} \top \Big)
        % \\ &
        \equiv \top \mathcal{S}_{[i_1, i_1]} \Big( A \mathcal{S}_{[ i_2 - i_1, i_2 - i_1]} \top \Big).
    \end{align*}
\end{theorem}

Following Theorems~\ref{th:AlwaysAfter2} and~\ref{th:AlwaysBefore2}, any negation-free formula in BMTL has an equivalent formula free of operators $\boxplus$ and $\boxminus$. Combined with Proposition~\ref{prop:diamond_si_ui}, a negation-free BMTL fragment can be defined using only $\since$ and $\until$; 

\begin{align*}
    A \coloneqq  p & | \top | A \land A | A \Si_{I} A | A \Ui_{I} A \\
    % & | \diamondplus_{[i1, i1]} (A \Ui_{[i2-i1,i2-i1]} \top) | \diamondminus_{[i1, i1]} (A \Si_{[i2-i1,i2-i1]} \top)
\end{align*}

In this fragment, the same temporal relations as in full MTL can be expressed. The only compromise is situated within the non-temporal relations; $A_1 \lor A_2$, often achieved via $\neg( \neg A_1 \land \neg A_2)$ can not be expressed in negation-free BMTL.

Theorems~\ref{th:AlwaysAfter2} and~\ref{th:AlwaysBefore2} swap out interval $[i_1, i_2]$ for two new intervals $[ i_1, i_1 ]$ and $[ i_2 - i_1, i_2 - i_1 ]$. Important to note is these newly introduced intervals are both singleton intervals.
% Furthermore, note that the proposed equivalencies make use of singleton intervals, i.e. intervals which only consist of a single value.
Both~\cite{alur_benefits_1996} and~\cite{brandt_querying_2018} point out that allowing these singleton intervals in our definitions may lead to counter-intuitive behaviour of the logic; 
in dense time domains, requiring an event to occur exactly at time 
$t$ (as opposed to within an interval) can make the satisfaction of a formula fragile: a minuscule perturbation in the timestamp can cause the formula to no longer hold. This is particularly problematic in real-world systems where timestamps may be imprecise or noisy. Moreover, in the context of model checking and satisfiability, allowing punctual constraints leads to undecidability, as shown in~\cite{alur_benefits_1996} and~\cite{ouaknine_some_2008}.
~\cite{alur_benefits_1996} introduces the singleton-free Metric Interval Logic (MITL) with the specific intent of obtaining a decidable fragment of MTL, following the observation that MTL is undecidable. 
Theorems~\ref{th:AlwaysAfter2} and~\ref{th:AlwaysBefore2} are therefore not applicable to MITL and its subfragments, for example CFMTL~\citep{ouaknine_some_2008}.

The following theorems offer alternative equivalences that do not employ singleton intervals. Naturally, the restriction increases the intricacy of the matter, resulting in equivalencies and proofs that stray further from intuition than the previous theorems.

The idea of Theorem~\ref{th:AlwaysAfter} is based on the following observation:  $\boxplus_{[i_1, i_2]} A$ at time $t$ means $A$ holds over interval $[t+i_1, t+i_2]$, meaning $A$ started somewhere before, or at the latest at, $t+i_1$ and end at the earliest at $t + i_2$. We split this expression into two parts, one that states $A$ starts holding at the latest at $t + i_1$ and another that states $A$ stops holding at the earliest at $t + i_2$, with the condition that these two parts must remain connected. For the first part, we postulate that $A$ needs to hold for at least a length $i_2 - i_1$, starting somewhere between $t + i_1 - \frac{i_2 - i_1}{2}$ and $t+i_1$, meaning $A$ will hold until at least $t + i_1 + \frac{i_2 - i_1}{2}$, the halfway point of the interval $[t+i_1, t+i_2]$. Likewise, the second part states $A$ should hold for $i_2 - i_1$ time points, since at least $t + i_1 + \frac{i_2 - i_1}{2}$. These two parts are illustrated in Figure~\ref{fig:example_timeline}.

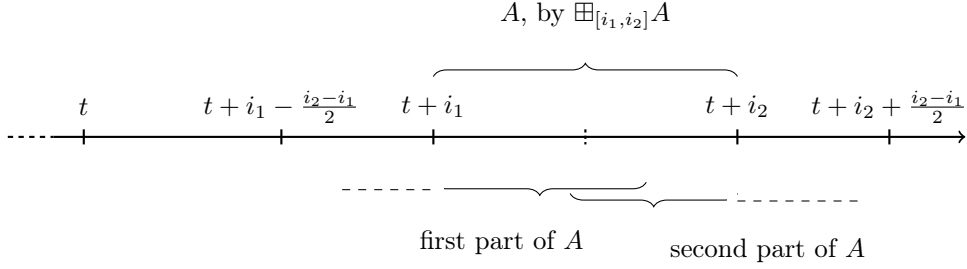
\begin{figure}[ht]
    \centering
    \begin{tikzpicture}[xscale=0.2, yscale=1]
        % Timeline with arrow
        \draw[thick, ->] (0, 0) -- (60,0);

        % Dashed segments
        \draw[thick,dashed, dash pattern=on 2pt off 2pt] (-3,0) -- (0,0);

        % Ticks and labels
        % \foreach \x/\label in {0/0, 3/2, 11/10, 14/12, 22/20, 25/22, 34/32, 38/36} {
        %     \draw[thick] (\x,0.1) -- (\x,-0.1);
        %     \node at (\x,0.4) {\label};
        % }

        \draw[thick] (2,0.1) -- (2,-0.1);
        \node at (2,0.4) {$t$};

        \draw[thick] (25,0.1) -- (25,-0.1);
        \node at (25,0.4) {$t + i_1$};
        \draw[thick] (45,0.1) -- (45,-0.1);
        \node at (45,0.4) {$t + i_2$};

        \draw[thick, dashed, dash pattern=on 1.2pt off 1.2pt] (35,0.1) -- (35,-0.1);

        \draw[thick] (15,0.1) -- (15,-0.1);
        \node at (15,0.4) {$t + i_1 - \frac{i_2 - i_1}{2}$};
        \draw[thick] (55,0.1) -- (55,-0.1);
        \node at (55,0.4) {$t + i_2 + \frac{i_2 - i_1}{2}$};

        % \draw[thick] (19,0) ellipse (0.3 and 0.15);
        % \node at (19,-0.4) {$t^-$};

        % \draw[thick] (54,0) ellipse (0.3 and 0.15);
        % \node at (54,-0.4) {$t^+$};

        \draw [decorate,decoration={brace,amplitude=5pt},yshift=+0.7cm]
            (25,0.1) -- (45,0.1) node [black,midway,yshift=+0.8cm] {$A$, by $\boxplus_{[i_1, i_2]} A$};
        
            % \draw [decorate,decoration={brace,amplitude=5pt, mirror},yshift=-0.7cm]
            %     (19,0.1) -- (39,0.1) node [black,midway,yshift=-0.8cm] {first part of $A$};
        
        % Dashed line from 19 to 25
        \draw [dashed]
            (19,-0.7) -- (25,-0.7);
    
        % Solid brace from 25 to 39 with label
        \draw [decorate,decoration={brace,amplitude=5pt,mirror},yshift=-0.7cm]
        (25,0.1) -- (39,0.1) node [black,midway,yshift=-0.8cm, xshift=-0.5cm] {first part of $A$};

        % White rectangle to hide the left curl
        \fill[white] (25,-0.2) rectangle (25.7,-1.0);

        % \draw [decorate,decoration={brace,amplitude=5pt, mirror},yshift=-0.85cm]
            % (34,0.1) -- (54,0.1) node [black,midway,yshift=-0.65cm] {second part of $A$};

        % Dashed line from 45 to 54
        \draw [dashed]
            (45,-0.85) -- (53,-0.85);
    
        % Solid brace from 34 to 45 with label
        \draw [decorate,decoration={brace,amplitude=5pt,mirror},yshift=-0.85cm]
        (34,0.1) -- (45,0.1) node [black,midway,yshift=-0.75cm, xshift=1.5cm] {second part of $A$};

        % White rectangle to hide the left curl
        \fill[white] (44.35,-0.2) rectangle (45,-1.0);
        
    \end{tikzpicture}
    \caption{Illustration of Proof~\ref{th:AlwaysAfter2}. The accolades of parts 1 and 2 have a set length of $i_2 - i_1$, but shift according to the exact location of $t^-$ and $t^+$. Regardless of the exact location of $t^-$ and $t^+$, the two accolades will always produce an area spanning at least $[t+i_1, t+i_2]$ in which $A$ holds everywhere.}
    \label{fig:example_timeline}
\end{figure}

In the following proof, we show these two parts translate to expressions without box operators or singleton intervals.

\begin{theorem}
    \label{th:AlwaysAfter}
    % Suppose $t\in T$, $I = [i_1, i_2]$ an interval, $\kappa, \lambda \in T_{\geq 0} $ and $A$ a formula. Then 

    Suppose $ t$ a time point, $I = [ i_1, i_2]$ a non-negative interval, $\kappa$ and $\lambda$ positive time points and $A$ a formula. Then 

    \begin{align*}
        \boxplus _{[i_1, i_2]} A \equiv 
        &\Big( \diamondplus_{[\frac{3i_1 - i_2}{2}, i_1]} \big( A \mathcal{U}_{[i_2 - i_1, i_2 - i_1 + \kappa]} \top \big) \Big) 
        % \\ &
        \land
        \Big(\diamondplus_{[i_2, \frac{3i_2 - i_1}{2}]} \big( A \mathcal{S}_{[i_2 - i_1, i_2 - i_1 + \lambda]} \top \big) \Big)
    \end{align*}

\end{theorem}

\begin{proof}
    
   Let $\kappa$ and $\lambda$ be arbitrary but fixed positive time points. We consider the following two formulae:

\begin{align}
\diamondplus_{\left[\frac{3i_1 - i_2}{2}, i_1\right]} \left( A \mathcal{U}_{[i_2 - i_1, i_2 - i_1 + \kappa]} \top \right) \label{eq:OnceUntil} \\
\diamondplus_{\left[i_2, \frac{3i_2 - i_1}{2}\right]} \left( A \mathcal{S}_{[i_2 - i_1, i_2 - i_1 + \lambda]} \top \right) \label{eq:OnceSince}
\end{align}

If for a given model $M$, formula~\eqref{eq:OnceUntil} holds at time $t$ in $M$, the semantics assert that there exists a time point within the interval $\left[ t + \frac{3  i_1 - i_2}{2}, t + i_1\right]$, say $t^-$, at which the formula $A \mathcal{U}_{[i_2 - i_1, i_2 - i_1 + \kappa]} \top$ holds. This means that, starting from $t^-$, $A$ must hold continuously for a duration of at least $i_2 - i_1$ time units, until a point is reached 
% (between $i_2 - i_1$ and $i_2 - i_1 +\kappa$)
where the trivially true formula $\top$ holds in $M$.

Dually, formula~\eqref{eq:OnceSince} captures a symmetric condition. It states that if the formula~\eqref{eq:OnceSince} holds in $M$ at time $t$, there exists a time point $t^+$ within the interval $\left[t + i_2, t + \frac{3i_2 - i_1}{2}\right]$ such that the formula $A \mathcal{S}_{[i_2 - i_1, i_2 - i_1 + \lambda]} \top$ holds in $M$. This implies that, looking backward from $t^+$, the proposition $A$ must have held continuously for a duration of at least $i_2 - i_1$, starting from a point where $\top$ is satisfied.

\vspace{0.5cm}

Observe that the lower bound of the interval $\left[ t + \frac{3  i_1 - i_2}{2}, t + i_1\right]$ can be rewritten as 

\begin{align*}
 t + \frac{3i_1 - i_2}{2} &= t + \frac{3i_1}{2} - \frac{i_2}{2} = t + i_1  + \frac{i_1}{2} - \frac{i_2}{2} \\
 &= t + i_1 + \frac{i_1 - i_2}{2} = t + i_1 - \frac{i_2 - i_1}{2} \\
\end{align*}

This latter expression identifies the time point that lies exactly one interval length of $i_2 - i_1$ before the halfway point of the interval $[t + i_1, t + i_2]$. 

From any time point within the range $[t + i_1 - \frac{i_2 - i_1}{2},\ t + i_1]$, the formula in Equation~\eqref{eq:OnceUntil} requires a forward temporal progression (a “jump”) of at least $i_2 - i_1$ time units. This jump begins no later than $t + i_1$ and, in the earliest case, lands at $t + i_1 - \frac{i_2 - i_1}{2} + (i_2 - i_1) = t + i_1 + \frac{i_2 - i_1}{2}$, i.e.~the midpoint of the interval $[t + i_1 , t + i_2]$. Since the proposition $A$ is required to hold throughout the entire duration of this jump, it follows that $A$ must hold continuously over at least the first half of the interval $[t+ i_1, t+ i_2]$, that is, from $t + i_1$ to $t + i_1 + \frac{i_2 - i_1}{2}$. It is important to note that this condition provides no information on the truth value of $A$ prior to $t + i_1$, nor beyond $t + i_1 + \frac{i_2 - i_1}{2}$. 
% The requirement is strictly confined to the temporal segment traversed by the jump.

\vspace{0.5cm}

The same method of interpretation applies to formula \eqref{eq:OnceSince}. In this case the upper bound interval can be similarly  rewritten, i.e. 

\begin{align*}
 t + \frac{3i_2 - i_1}{2} 
 &= t + \frac{3i_2}{2} - \frac{i_1}{2} 
 = t + i_2  + \frac{i_2}{2} - \frac{i_1}{2} = t + i_2 + \frac{i_2 - i_1}{2} \\
\end{align*}

From any time point within the range $[t + i_2, t + i_2 + \frac{i_2 - i_1}{2}]$, the formula in Equation~\eqref{eq:OnceSince} requires a backwards ``jump'' of at least $i_2 - i_1$ time units. As such, this jump back ends at the latest at $t + i_2$ and starts at the latest at $t + i_2 + \frac{i_2 - i_1}{2} - (i_2 - i_1) = t + i_2 - \frac{i_2 - i_1}{2}$, i.e.~the midpoint of the interval $[t+ i_1, t+ i_2]$. Again, A must continuously hold during this jump. As such, $A$ holds continuously  from $t + i_2 - \frac{i_2 - i_1}{2}$ to $t + i_2$, which is the latter half of the interval $[t+ i_1, t+ i_2]$. 

\vspace{0.5cm}

As a result, if the the conjunction of~\eqref{eq:OnceUntil} and~\eqref{eq:OnceSince} holds for an model $M$ at time point $t$, then $A$ holds in $M$ over (at least) the interval $[ t + i_1 , t+ i_2]$, i.e. $M, t \models \boxplus_{[i_1, i_2]} A$.

\vspace{0.5cm}

Suppose that $M, t \models \boxplus_{[i_1, i_2]} A$. By the semantics of the $\boxplus$ operator, $A$ holds at every time point within the interval $[t + i_1, t + i_2]$
% , i.e., $\forall t'' \in [t + i_1, t + i_2] : M, t'' \models A$
    . Furthermore, since $\top$ is valid at all time points, it trivially holds that $M, t + i_2 \models \top$. Observe that $t + i_2$ lies within the interval $[t + i_2, t + i_2 + \kappa]$, and thus the condition $A \mathcal{U}_{[i_2 - i_1, i_2 - i_1 + \kappa]} \top$ is satisfied at time $t + i_1$. Moreover, since $t + i_1$ belongs to the interval $[t + \frac{3i_1 - i_2}{2}, t + i_1]$, it follows that formula~\eqref{eq:OnceUntil} holds at time $t$.

\vspace{0.5cm}

We now turn our attention to Equation~\eqref{eq:OnceSince}. Note that $t + i_1$ lies within the interval $[t + i_1 - \lambda, t + i_1]$, which can be equivalently expressed as $[(t + i_2) - (i_2 - i_1 + \lambda), (t + i_2) - (i_2 - i_1)]$. Given that $A$ holds throughout $[t + i_1, t + i_2]$ and $\top$ holds at $t + i_2$, it follows that $M, t + i_2 \models A \mathcal{S}_{[i_2 - i_1, i_2 - i_1 + \lambda]} \top$, and consequently formula~\eqref{eq:OnceSince} holds in $M$ at time $t$. Since $t + i_2$ lies within the interval $[t + i_2, t + \frac{3i_2 - i_1}{2}]$, we conclude that there exists a time point $t^+$ in the interval $[t+i_2, t + \frac{3i_2 - i_1}{2}]$ for which $A \Si_{[i_2 - i_1, i_2 -i_1 +\lambda]} \top$ holds in $M$ at time $t^+$. As a result, formula~\eqref{eq:OnceSince} holds in $M$ at time point $t$.

\vspace{0.5cm}

$M, t \models \boxplus_{[i_1, i_2]} A$ thus induces that both formulae~\eqref{eq:OnceSince} and~\eqref{eq:OnceUntil} hold in $M$ at time $t$, and therefore their conjunction as well.
    
\end{proof}

Analogously, $\boxminus_{[i_1, i_2]} A$ can be rewritten into a similar expression.

\begin{theorem}
    \label{th:AlwaysBefore}
    Suppose $t\in T$, $I = [i_1, i_2]$ a non-negative interval, $\kappa, \lambda$ positive time points and $A$ a formula. Then

    \begin{align*}
         \boxminus _{[i_1, i_2]} A \equiv 
        &\Big( \diamondminus_{[\frac{3i_1 - i_2}{2}, i_1]} \big( A \mathcal{S}_{[i_2 - i_1, i_2 - i_1 + \lambda]} \top \big) \Big) 
        % \\ &
        \land \Big(\diamondminus_{[i_2, \frac{3i_2 - i_1}{2}]} \big( A \mathcal{U}_{[i_2 - i_1, i_2 - i_1 + \kappa]} \top \big) \Big)
    \end{align*}
\end{theorem}

\begin{proof}
    The proof of Theorem \ref{th:AlwaysBefore} is analogous to that of Theorem \ref{th:AlwaysAfter} 
    % and is left as an exercise for the reader
    .
\end{proof}

It follows from Theorems~\ref{th:AlwaysAfter} and~\ref{th:AlwaysBefore} that any formula with bounded interval subscripts can be rewritten into formulae containing only the temporal operators $\diamondplus$, $\diamondminus$, $\Si$ and $\Ui$. Combined with Proposition~\ref{prop:diamond_si_ui}, we conclude that negation-free Bounded MTL and Bounded MITL can be defined using only temporal operators $\since$ and $\until$ without compromising on temporal expressivity compared to the full Bounded MTL and Bounded MITL respectively.

% \section{Discussion}\label{sec:Discussion}
% \input{TempDatalog/Discussion.tex}

\section{Discussion and conclusion}\label{sec:conclusion}

In this work, we have shown that the syntax of negation-free {MTL} with bounded intervals can be significantly simplified without compromising its expressive power. Specifically, we demonstrated that all occurrences of the temporal operators $\boxplus$ and $\boxminus$ can be syntactically eliminated in favor of combinations of $\diamondplus$, $\diamondminus$, $\mathcal{U}$, and $\mathcal{S}$. This is the case for negation-free BMTL, as follows from Theorems~\ref{th:AlwaysAfter2}-\ref{th:AlwaysBefore2}, but also for the more restrictive Bounded MITL, as illustrated by Theorems~\ref{th:AlwaysAfter}-\ref{th:AlwaysBefore}. Furthermore, we established that even the ``once'' operators $\diamondplus$ and $\diamondminus$ can be removed, yielding a fragment that relies solely on the $\Si$ and $\Ui$ operators. 

This result is particularly striking, given the absence of negation in the fragment under consideration. While classical transformations of ``always'' operators into ``once'' and negation are well-known in temporal logic, our findings reveal that similar reductions are possible even in negation-free settings. This challenges the common intuition that negation is essential for expressing ``always'' expressions. Given the influence of MTL on recent temporal reasoning frameworks, the obtained results can be carried over to MTL-based frameworks such as MEL and DatalogMTL.

Looking forward, several avenues for future research emerge.
From a theoretical perspective,  
it is an open question whether analogous syntactic reductions can be achieved in temporal reasoning frameworks with different temporal operators or semantics, for example LARS. This would prompt us to reconsider the theoretical (minimal) requirements of temporal logic frameworks for open-world infrastructures such as IoT and the Semantic Web. These theoretical results help us better understand the strengths and weaknesses of stream reasoning systems, which in turn allows us to better target research efforts towards more efficient, scalable, and predictable reasoning systems for dynamic and data-intensive environments.
 
Approached from a more practical angle, the expressions obtained in this paper can be leveraged in MTL-based frameworks; via the obtained expressions, the more complex temporal operators $\boxplus, \boxminus, \diamondplus, \diamondminus$ can be implemented as ``built-ins'' on top of the existing theory, to improve readability (and thereby user convenience) without intervention in the underlying theory. Finally, an interesting direction for future work is to explore whether similar syntactic simplifications can be extended to unbounded intervals (i.e., $[0, +\infty[$). While this would likely require a different methodological approach, it opens up promising possibilities for broadening the applicability of the results presented here.

\section{Acknowledgements}
This research is funded by the FRACTION project (Nr. G086822N), funded by the Fonds voor Wetenschappelijk Onderzoek (FWO) organization.

\newpage
% \appto{\bibsetup}{\sloppy}
% \printbibliography

\bibliography{biblio.bib}

\end{document}